\documentclass[a4paper, 11pt]{amsart}

\usepackage{amsmath, amssymb, amsthm}
\usepackage[english]{babel}
\usepackage[T1]{fontenc}
\usepackage{enumerate}
\usepackage{graphicx}
\usepackage{rotating}
\usepackage{array,multirow}
\newcommand{\ee}{\mathbf{e}}

\def\CI {\mathcal{I}}

\def\eps {\epsilon}

\def \be {  \varpi}

\newcommand{\fer}[1]{(\ref{#1})}

\newcommand{\ave}[1]{\langle#1\rangle}

\newcommand{\R}{\mathbb R}
\def\be#1\ee{\begin{equation}#1\end{equation}}

\numberwithin{equation}{section}

\newcommand{\bq}{\begin{equation}}
\newcommand{\eq}{\end{equation}}

\newtheorem{thm}{Theorem}

\theoremstyle{remark}
\newtheorem{rem}{Remark}
\theoremstyle{definition}

\begin{document}

\title[Fokker--Planck equations for opinion formation]{WRIGHT--FISHER--TYPE EQUATIONS  FOR OPINION FORMATION, LARGE TIME BEHAVIOR AND WEIGHTED logarithmic-SobOLEV INEQUALITIES}

\author{GIULIA FURIOLI}
\address{DIGIP, University of Bergamo, viale Marconi 5, 24044 Dalmine, Italy}
\email{giulia.furioli@unibg.it} 

\author{ADA PULVIRENTI}
\address{Department of Mathematics, University of Pavia, 
via Ferrata 1,
Pavia, 27100 Italy}
\email{ada.pulvirenti@unipv.it}

\author{ELIDE TERRANEO}
\address{Department of Mathematics,
University of Milan, via Saldini 50, 20133 Milano, Italy }
\email{elide.terraneo@unimi.it}

\author{GIUSEPPE TOSCANI}
\address{Department of Mathematics, University of Pavia,
via Ferrata 1,
Pavia, 27100 Italy}
\email{giuseppe.toscani@unipv.it}

\maketitle

\begin{abstract}
We study the rate of convergence to equilibrium of the solution of a Fokker--Planck type equation introduced in \cite{Tos06} to describe opinion formation in a multi-agent system. The main feature of this Fokker--Planck equation is the presence of a variable diffusion coefficient and boundaries, which introduce new challenging mathematical problems in the study of its long-time behavior. 
\end{abstract}


\section{Introduction}\label{intro}
Kinetic models for (continuous) opinion formation have been first introduced and discussed in \cite{Tos06}, starting from the study of a multi-agent system in which agents undergo binary interactions so that the personal opinion could be changed by means of compromise and self-thinking \cite{BKR03,BKVR03,FPTT17}. 
In most of the problems related to socio-economic studies of multi-agent systems  \cite{NPT,PT13}, the variable is assumed to vary in an unbounded domain (mainly the positive half-line). On the contrary, the opinion variable is assumed to take values in  the bounded interval $\CI= (-1, 1)$, the values $\pm 1$ denoting the extremal opinions. 
Among the various models introduced in \cite{Tos06} (cf. also \cite{Bou,DMPW}), one Fokker--Planck type equation has to be distinguished in view of its equilibrium configurations, which are represented by Beta-type probability densities supported in the interval $(-1, 1)$. This Fokker--Planck equation for the opinion density $v(t,y)$, with $|y| < 1$,  is given by 
 \be\label{op-FP}
 \frac{\partial v(t,y)}{\partial t} = \frac \lambda 2\frac{\partial^2 }{\partial y^2}\left((1-y^2)
 v(t,y)\right) + \frac{\partial }{\partial y}\left((y -m)v(t,y)\right).
 \ee
 In \fer{op-FP}, $\lambda$ and $m$ are given constants, with   $\lambda >0$ and $-1 <m<1$. 
Suitable boundary conditions at the boundary points $y = \pm 1$ then guarantee conservation of mass and momentum of the solution \cite{FPTT17}.  
 Equation \fer{op-FP} possesses  steady states which solve
\[
 \frac
\lambda 2\frac{d}{d y}\left((1-y^2)v(y)\right) + (y -m) v(y)= 0.
 \]
In case a mass density equal to unity is chosen, the steady state equals a probability density of Beta type, given by
 \be\label{beta}
v_{m,\lambda}(y)= C_{m,\lambda} (1-y)^{-1 + \frac{1-m}\lambda} (1+y)^{-1 + \frac{1+m}\lambda}.
 \ee
In \fer{beta} the constant $C_{m,\lambda}$ is such that the mass of $v_{m,\lambda}$
is equal to one. Since $-1 <m<1$, $v_{m,\lambda}$ is integrable on $\CI$.
Note that $v_{m,\lambda}$  is continuous on $\CI$, and as soon as $\lambda > 1+|m|$ tends to infinity as $y \to \pm 1$. 

A better understanding of the social meaning of the parameters $\lambda$ and $m$ appearing in \fer{op-FP} comes from the microscopic description of the opinion change in a multi-agent system through binary interactions among agents, leading to the Boltzmann type kinetic equation considered in \cite{Tos06}. 
Given a pair of agents with opinions $x$ and $x_\ast$, it was assumed in \cite{Tos06} that any elementary interaction between them modifies the entering opinions according to 
\begin{equation}
\begin{split}
    x'&=x+\gamma (x_\ast-x)+D(x)\eta, \\
    x_\ast'&=x_\ast+\gamma (x-x_\ast)+D(x_\ast)\eta_\ast.
    \label{eq:binary}
\end{split}
\end{equation}
The right-hand side of  \fer{eq:binary} describes the modification of the opinion in terms of the quantity $\gamma(x_\ast-x)$ (respectively $\gamma(x_\ast-x)$), that measures the \emph{compromise} between opinions with intensity $\gamma$,  $0<\gamma<1$,  and a random contribution, given by  the  random variable $\eta$ (respectively $\eta_\ast$), modelling stochastic fluctuations induced by the \emph{self-thinking} of the agents.  $D(\cdot)\geq 0$ is an opinion-dependent diffusion coefficient modulating the amplitude of the stochastic fluctuations, that is the variance of $\eta$ and $\eta_\ast$. In \cite{Tos06} the two random variables were assumed to be independent and identically distributed with zero mean and variance $\sigma^2$.   Let us further set
 \be\label{lam}
 \lambda = \frac{\sigma^2}\gamma.
 \ee
Then, interactions of type \fer{eq:binary} with small values of $\lambda$ characterize compromise dominated societies, while interactions with large values of $\lambda$ characterize self-thinking dominated societies. 

Introducing the distribution function $f=f(t,\,x):\R_+\times [-1,\,1]\to\R_+$, such that $f(t,\,x)dx$ is the fraction of agents with opinion in $[x,\,x+dx]$ at time $t$, the binary rules~\eqref{eq:binary} give rise to a Boltzmann-type kinetic equation, that in weak form reads
\begin{multline}
    \frac{d}{dt}\int_{-1}^1\varphi(x)f(t,\,x)\,dx \\
    =\frac{1}{2}\int_{-1}^1\int_{-1}^1\ave{\varphi(x')+\varphi(x_\ast^\prime)-\varphi(x)-\varphi(x_\ast)}f(t,\,x)f(t,\,x_\ast)\,dx\,dx_\ast,
    \label{eq:boltz}
\end{multline}
where $\varphi:[-1,\,1]\to\R$ is an arbitrary test function, i.e. any observable quantity depending on the microscopic state of the agents, and where we denoted by $\langle \cdot \rangle$ the mathematical expectation.  Choosing $\varphi(x)=1$, one shows that the integral of $f$ with respect to $x$ is constant in time, i.e. that the total number of agents is conserved. This also implies that $f$ can be thought as a probability density for every $t>0$. Choosing instead $\varphi(x)=x$,  and considering that \fer{eq:binary} implies
 \[ 
 \langle x^\prime+x^\prime_\ast \rangle = x + x_\ast,
 \] 
one concludes that
\begin{equation}
    \frac{d}{dt}\int_{-1}^1 xf(t,\,x)\,dx= 0.
         \label{eq:mean}
\end{equation}
Therefore the mean opinion $m:=\int_{-1}^1 xf(t,\,x)\,dx$ is  conserved in time.
  As shown in \cite{Tos06}, one can recover an explicit expression of the asymptotic distribution function at least in the so-called \textit{quasi-invariant regime}, i.e. the one in which the variation of the opinion in each binary interaction is small. To describe such a regime, one scales the parameters $\gamma$, $\sigma^2$ in~\eqref{eq:binary} as
\begin{equation}
    \gamma\to\epsilon\gamma, \qquad \sigma^2\to\epsilon\sigma^2,
    \label{eq:scaling}
\end{equation}
where $\epsilon>0$ is an arbitrarily small scaling coefficient. Moreover,  to study the large time behavior of the system, one introduces the new time scale $t \to\epsilon t$ and  scales the distribution function as $v(t,\,x):=f(\frac{t}{\epsilon},\,x)$. In this way,  at every fixed $t>0$ and in the limit $\epsilon\to 0^+$, $v$ describes the large time trend of $f$. Moreover, as shown in \cite{Tos06}, if $D(x) = \sqrt{1-x^2}$, $v(t,x)$ satisfies the Fokker--Planck equation \fer{op-FP}.

Since the value of $\lambda$ is left unchanged by the scaling \fer{eq:scaling} leading from the Boltzmann-type equation \fer{eq:boltz} to the Fokker--Planck type equation \fer{op-FP}, the parameter $\lambda$ maintains its meaning also in the target equation.  The roles of the constants $\lambda$ and $m$ are evident also by looking at the shape of the steady Beta distribution \fer{beta}.  We can observe that, by fixing for example $m>0$, increasing the values of $\lambda$, and consequently moving from a compromise dominated to a self-thinking dominated society, such a distribution may depict a transition from a strong consensus around the mean  to a milder consensus, and further to a radicalisation in the extreme opinion $x=1$  up to the appearance of a double radicalisation in the two opposite extreme opinions $x=\pm 1$.

In view of the described social meaning, a relevant problem related to the solution to the Fokker--Planck equation \fer{op-FP} is to understand at which speed the solution $v(t)$ converges to its equilibrium configuration, and to reckon how this rate depends on the parameters $\lambda$ and $m$. Indeed, as outlined before, it is easily recognized that different values of these parameters give raise to situations in which the extremal opinions are not attracting, and this happens for $\lambda < 1-|m|$, or situations in which opinions are polarized around the extreme ones ($\lambda > 1+|m|$). Also, it is not clear if these (different) steady states are reached very quickly in time, independently of the values of the parameters.

As discussed in \cite{FPTT17}, in analogy with  the methods developed for the classical Fokker--Planck equation \cite{T99}, the large-time behavior of the solution to \fer{op-FP} can be fruitfully studied by resorting to entropy methods \cite{MTV}. 
This corresponds to the study of the evolution in time of various Lyapunov functionals, the most known being the Shannon entropy of the solution relative to the steady state.  We recall here that the relative Shannon entropy of two probability densities $f$ and $g$ supported on the bounded interval $\CI$ is {defined} by the formula
 \be\label{relH}
 H(f,g)= \int_{\CI} f(x) \log \frac {f(x) }{g(x) }\, dx.
 \ee
 Note that $H(f,g)$ can be alternatively written as
\[
\int_\CI \left( \frac{f(x)}{g(x)} \log \frac{f(x)}{g(x)} -  \frac{f(x)}{g(x)} +1\right ) g(x)dx,
\]
which is the integral of a nonnegative function.

As shown in \cite{FPTT17}, the relative entropy  $H(v(t), v_{m,\lambda})$ decreases in time, and its time variation can be expressed by the \emph{entropy production} term
 \be\label{ep}
\tilde I(v(t), v_{m,\lambda}) = 
\int_ {\CI}\frac \lambda 2 (1-y^2)
\left(\partial_y \log  \frac {v(t,y)}{ v_{m,\lambda}(y)}\right )^2 v(t,y) dy.
 \ee 
While for the classical Fokker--Planck equation \cite{T99},  exponential in time convergence at explicit rate follows in consequence of the logarithmic Sobolev inequality,  the results in presence of the weight in \fer{ep} are less satisfactory. Various convergence results have been obtained in \cite{FPTT17} by resorting to a generalization of the so-called Chernoff inequality with weight, first proven by Klaassen \cite{Kla}. The main consequence of this inequality \cite{FPTT17}, was to show that exponential convergence to equilibrium with an explicit rate holds at least for initial values $v_0$ for \fer{op-FP}  close to the steady state \fer{beta} in the weighted $L^2$-norm
\be\label{L2}
\|v_0-v_{m,\lambda}\|_*^2 := \int_\CI |v_0(y)-v_{m,\lambda}(y)|^2 v_{m,\lambda}(y)^{-1}\, dy.
\ee
Also, a weaker convergence result was proven for general initial data, by showing that the standard $L^1$-distance decays to zero at a polynomial rate (without any explicit rate of convergence). 

Related results have been obtained by Epstein and Mazzeo in \cite{EM10} for the adjoint equation 
\be\label{ad-FP}
 \frac{\partial u(t,x)}{\partial t} = \frac \lambda 2 (1-x^2) \frac{\partial^2  u(t,x) }{\partial x^2}- (x-m)\frac{\partial  u(t,x) }{\partial x} , \quad  t>0,\quad x \in \CI.
\ee
Indeed, the Fokker--Planck equation  \fer{op-FP} is naturally coupled to   \fer{ad-FP} since, at least formally, if $v$ is a solution of  \fer{op-FP}, then 
 \be\label{rel1}
 u(t,x) = \frac{v(t,x)}{v_{m,\lambda}(x)}
 \ee
is a solution of \fer{ad-FP} (remark that the notation we have chosen for the solutions $v(t,y)$ of  \fer{op-FP} and $u(t,x)$ of \eqref{ad-FP} is the same as in the paper  \cite{EM10} to which we will often refer in the sequel of the paper).
 Among other results, in \cite{EM10} exponential convergence in $L^1(\CI)$ of $v(t)$ towards $v_{m,\lambda}$ has been proven (without rate) by resorting to classical analysis of semigroups.
 
In this paper we aim at proving that entropy methods can also produce exponential convergence in $L^1(\CI)$ towards equilibrium with an explicit rate, at least in some range of the parameters $\lambda$ and $m$.
The result follows from a new weighted logarithmic-Sobolev inequality satisfied by the Beta functions \fer{beta} when they belong to $L^2(\CI)$. In this case, we will prove  that there exists an explicitly computable constant $K_{m,\lambda} >0$ such that,  for any
  probability density $\varphi \in L^1(\CI)$ absolutely continuous with respect to $v_{m,\lambda}$ 
  \be\label{ok}
 H(\varphi, v_{m,\lambda}) \leq K_{m,\lambda} \tilde I(\varphi, v_{m,\lambda}).
 \ee
 Inequality \fer{ok} requires that $\lambda >0$, $m\in \CI$ be such that
 \[
1-\frac \lambda 2 >0, \quad {\rm if}\,\, m=0, \quad
1-\frac \lambda 2 \geq |m|, \quad {\rm if}\,\, m\neq 0.
\] 
and  allows us  to obtain exponential convergence in relative entropy with an explicitly computable rate.

In more details, this is the plan of the paper: we will start by recalling in Section \ref{sol} an existence result for the initial-boundary value problem for the Fokker--Planck equation \fer{op-FP}, as follows from the analysis of Wright--Fisher type equations presented in \cite{EM10} for the adjoint equation \fer{ad-FP}. Then, the proof of the new logarithmic-Sobolev inequality for Beta functions and its consequences on the large-time behavior of the solution to equation \fer{op-FP} will be studied in Section \ref{LS}. 
Last, in Sections \ref{dist} and \ref{concl} we will discuss the case $m=0$, $\lambda =1$ which leads to a uniform density at equilibrium,  and we will address some concluding remarks.

\section{Existence and properties of solutions}\label{sol}
For given constants $\lambda >0$ and $m\in \CI$, let us consider the initial-boundary value problem
\be\label{main}
\left\{
\begin{aligned}
&\partial_t v(t,y)= \frac \lambda 2\, \partial_y^2 \left((1-y^2) v(t,y)\right ) +\partial_y \left((y-m)v(t,y)\right ),\quad  t>0,\quad y \in \CI\\
&v(0,y)=v_0(y) \ge 0 \in L^1(\CI), 
\end{aligned}
\right .
\ee
with boundary conditions
\be\label{bc-mom}
\lim_{y\to -1^+} (1-y^2)v(t,y)= \lim_{y\to 1^-} (1-y^2)v(t,y)= 0,  \quad t>0 
\ee
and 
 \be\label{bc-mass}
 \left\{
 \begin{aligned}
 &\lim_{y\to -1^+} (y-m)v(t,y) + \frac\lambda{2} \frac{\partial}{\partial y}\left((1-y^2) v(t,y)\right) = 0,  \quad t>0\\
 &\lim_{y\to 1^-} (y-m)v(t,y) + \frac\lambda{2} \frac{\partial}{\partial y}\left((1-y^2) v(t,y)\right)  = 0,\quad t>0 .
 \end{aligned}
 \right .
 \ee
Conditions \eqref{bc-mom} and \eqref{bc-mass}  are suggested by the nature of the problem, since they imply momentum and mass conservation of the (possible) solution to the Fokker--Planck equation. 
While condition  \fer{bc-mom} is automatically satisfied for a sufficiently
regular density $v$, condition  \fer{bc-mass} requires an exact balance
between the so-called advective and diffusive fluxes on the boundaries
$y= \pm 1$. This condition is usually referred to as the \emph{no-flux} boundary
condition \cite{FPTT17}.

The linear Fokker--Planck equation in \fer{main} has a variable diffusion coefficient and the variable $y$ belongs to the bounded interval $\CI$, and this requires to consider boundary conditions. An alternative formulation would be to  consider the pure initial value problem on the whole real line, by introducing the diffusion coefficient $(1-y^2) \chi(\CI)$, where $\chi(X)$ denotes the characteristic function of the set $X\subseteq \R$. The initial value problem for Fokker--Planck equations with general non smooth coefficients has been recently considered by Le Bris and Lions \cite{LL08}. However, diffusion coefficients as $(1-y^2) \chi(\CI)$ are not included in their analysis, and the results in \cite{LL08} do not apply. For such a problem a general theory about existence, uniqueness and continuous dependence on initial data  still does not exist. 

On the other hand, a quite general theory has been recently developed by  Epstein and Mazzeo in \cite{EM10} for the  equation \fer{ad-FP}.  Their results give some insight also on our Fokker--Planck equation \fer{op-FP}, subject to no-flux boundary conditions as given in \fer{bc-mass}.

Equation \fer{ad-FP} is a Wright--Fisher type equation, of the form
\[
\partial_t u(t,x)= a(x) \partial_x^2 u(t,x) +b(x) \partial_x u(t,x),\quad t>0,\quad  x \in (A,B)
\]
where $A$, $B \in \R$, $a\in C^\infty([A,B])$,  $b \in C^\infty([A,B])$  with
\[
a(x)=(x-A)(B-x)\tilde a(x),\quad \tilde a \in C^\infty([A,B]), \quad \tilde a(x) >0 \text{\  for all \ } x\in [A,B],
\]
and
\[
b(A)\geq 0,\quad b(B)\leq 0.
\]
Since our results heavily depend on the precise analysis by Epstein and Mazzeo on the solutions of the Wright--Fisher--type equations, we collect in the next Theorem the results we need about these solutions. All the details can be  extracted from \cite{EM10}. In the rest, we will use as usual the notation $\bar\CI = [-1,1]$.

\begin{thm}[Epstein--Mazzeo \cite{EM10}] \label{EM}
For all constants $\lambda >0$ and $m\in \CI$ let us consider the initial-boundary value problem \fer{main}
with no-flux boundary conditions, as given by \fer{bc-mass}. 
Then, there exists a kernel $q_t(x,y):\{t>0, x\in \bar\CI, y\in \CI\} \rightarrow \mathbb R$ such that 
\be\label{Qt}
Q_tv_0(y):= \int_{-1}^1 q_t(x,y) v_0(x) dx
\ee
is a classical solution of the Cauchy problem. 
The kernel $q_t(x,y)$ satisfies the properties
\begin{enumerate}[1)]
\item $q_t(x,y) \in C^\infty\left((0,\infty)\times \bar\CI\times \CI\right)$;
\item $q_t(x,y) >0$ on $(0,\infty)\times \bar\CI\times \CI$;
\item for $y\to -1^+$ we have  $q_t (x,y) \sim (1+y)^{-1+\frac {1+m}\lambda } \varphi(t,x)$ for all $t>0$, $x\in \bar\CI$ with $\varphi \in C^\infty$;
\item for $y\to 1^-$ we have $q_t (x,y) \sim (1-y)^{-1+\frac {1-m}\lambda } \tilde \varphi(t,x)$ for all $t>0$, $x\in \bar\CI$ with $\tilde \varphi \in C^\infty$;
\item for all $t>0$ and all $x\in \bar\CI$ we have
\[
\begin{aligned}
&\lim_{y\to -1^+} \left(\frac \lambda 2 \partial_y \left((1-y^2) q_t(x,y)\right ) +(y-m)q_t(x,y)\right )=0\\
&\lim_{y\to 1^-} \left(\frac \lambda 2 \partial_y \left((1-y^2) q_t(x,y)\right ) +(y-m)q_t(x,y)\right )=0.
\end{aligned}
\]
\end{enumerate}
As a consequence, the solution $v(t,y)= Q_tv_0(y)$ satisfies
\begin{enumerate}[1')]
\item $v(t,y)\in C^\infty\left((0,\infty)\times \CI\right)$;
\item $v(t,y) >0$ on $(0,\infty)\times \CI$;
\item for $y\to -1^+$ we have  $v(t,y) \sim (1+y)^{-1+\frac {1+m}\lambda } \psi(t)$ for all $t>0$ with $\psi\in C^\infty$;
\item for $y\to 1^-$ we have $v(t,y) \sim (1-y)^{-1+\frac {1-m}\lambda } \tilde \psi(t)$ for all $t>0$ with $\tilde \psi \in C^\infty$;
\item for all $t>0$  we have (no flux boundary conditions)
\[
\begin{aligned}
&\lim_{y\to -1^+} \left(\frac \lambda 2 \partial_y \left((1-y^2) v(t,y)\right ) +(y-m)v(t,y)\right )=0\\
&\lim_{y\to 1^-} \left(\frac \lambda 2 \partial_y \left((1-y^2) v(t,y)\right ) +(y-m)v(t,y)\right )=0.
\end{aligned}
\]
\end{enumerate}
Moreover,  $v\in C((0,\infty), L^1(\CI))$ and
\[ 
\lim_{t\to 0^+} \|v(t) -v_0\|_{L^1} =0.
\]
\end{thm}
In consequence of the validity of no-flux boundary conditions (property \emph{5')}) conservation of mass follows. Hence, since $v_0$ is a probability density, the solution $v(t)= Q_tv_0$ remains a  probability density for all $t>0$. Indeed
\[
\begin{aligned}
\frac d{dt} \int_{-1}^1 v(t,y) dy & =  \int_{-1}^1 \partial_t v(t,y) dy =  \int_{-1}^1  \partial_y  \left(\frac \lambda 2 \partial_y \left((1-y^2) v(t,y)\right ) +(y-m)v(t,y)\right )dy\\
& =  \left[\frac \lambda 2 \partial_y \left((1-y^2) v(t,y)\right ) +(y-m)v(t,y)\right ]_{-1}^1 =0.
\end{aligned}
\]
The  steady states for equation \eqref{main} are given by the Beta densities \fer{beta}.

Some remarks are in order. First of all, by means of  $\emph{3')}$ and $\emph{4')}$ of Theorem \ref{EM} we conclude that,  for any given initial datum $v_0$ that is a probability density,  the solution $v(t)=Q_tv_0$ has the same behavior at the boundary of $\CI$ of the corresponding steady state $v_{m,\lambda}$.

Consequently, in reason of  the regularity of both functions,  the probability density  $v(t)$, solution of the initial value problem, is absolutely continuous with respect to the steady state $v_{m,\lambda}$ for all times $t >0$,  
\be\label{ac}
\frac {v(t)}{ v_{m,\lambda}} \in C_b^\infty(\CI)
\ee
and it can be continuously extended to $\bar\CI$.
In addition, if the condition 
\be\label{cond1}
1-\lambda >|m|
\ee
is satisfied, both the steady state and the solution $v(t)$ vanish on the boundary of the domain.

\section{Weighted logarithmic-Sobolev inequalities and large time behavior.}\label{LS}

As briefly discussed in the Introduction, our main goal is concerned with the study of the large-time behavior of the solution to the Fokker--Planck equation  \fer{op-FP}. This problem has been considered by Epstein and Mazzeo \cite{EM10},  who studied the large-time behavior of equation \fer{ad-FP}, and used this to prove exponential convergence in $L^1$ for large times of the solution $v(t)=Q_t v_0$ of the Cauchy problem \eqref{main} to the corresponding steady state $v_{m, \lambda}$  for the whole range of the allowed parameters $m\in \CI$ and $\lambda >0$.
While their result, obtained by classical semigroup arguments is very general,  the rate of the exponential convergence was not explicitly computed. A stronger result was recently obtained in \cite{FPTT17}. This result has been shown to hold for a large class of Fokker--Planck equations with non constant diffusion coefficients and bounded domains, by resorting to classical entropy type inequalities. 
Different Lyapunov functionals can be actually evaluated along the solution of the Fokker--Planck equation \fer{op-FP} and, in presence of some regularity of the solution itself, can be proven to be monotone decreasing in time.
Among them, the relative Shannon entropy defined in \fer{relH}, the Hellinger distance, the reverse relative Shannon entropy, and the weighted $L^2$-distance. 

Thanks to  Theorem \ref{EM}, we know that the solution of the o\-pi\-nion for\-mation equation \eqref{main} fulfills the conditions which allow the application of the formal results contained in \cite{FPTT17}. 
In particular,  the following result  about exponential convergence to equilibrium follows.
\begin{thm}[\cite{FPTT17}]\label{l2}
Let  $\lambda >0$  and $m\in \CI$ .  Let $v_0$ a probability density satisfying 
 \be\label{vic}
 \|v_0-v_{m,\lambda} \|_*^2 = \int_\CI\frac {\left(v_0(y)-v_{m,\lambda}(y)\right )^2}{v_{m,\lambda}(y)} dy <\infty
 \ee
where $v_{m,\lambda}$ is the stationary solution \eqref{beta} of  the Fokker--Planck equation  \eqref{main}. Then,  the solution $v(t,y)= Q_tv_0(y)$ of  \eqref{main} defined in \eqref{Qt} converges exponentially in time towards the steady state, and the following holds true
 \be\label{L22}
  \|v(t)-v_{m,\lambda} \|_*^2  \leq e^{-2t} \|v_0-v_{m,\lambda} \|_*^2  , \quad t >0.
 \ee
\end{thm}
Inequality \fer{L22} implies exponential convergence in $L^1$. Indeed  by Cauchy--Schwartz inequality, for any pair $f$, $g$ of probability  densities   on $\CI$ it holds
\[
\begin{aligned}
\int_\CI |f(y)-g(y)|\, dy &= \int_\CI \frac{|f(y)-g(y)| }{\sqrt{v_{m,\lambda}(y)} }\sqrt{v_{m,\lambda}(y)}\, dy\\
&\leq  
\left( \int_\CI \frac {\left(f(y)-g(y)\right )^2 }{v_{m,\lambda}(y)} \, dy \right )^{\frac 12}\left( \int_\CI  {v_{m,\lambda}}(y) \, dy\right )^{\frac 12}\\
&\leq  
\left( \int_\CI \frac {\left(f(y)-g(y)\right )^2 }{v_{m,\lambda}(y)} \, dy \right )^{\frac 12}.
\end{aligned}
\]
Hence, \eqref{L22} implies
\be\label{conv-expL1-L2}
\left \| v(t)-v_{m,\lambda}\right \|_{L^1} \leq e^{-t} \left( \int_\CI \frac {(v_0(y)-v_{m,\lambda}(y))^2}{v_{m,\lambda}(y)} dy\right )^{\frac 12}
\ee
for the whole set of allowed parameters $m\in \CI$ and $\lambda >0$. 

It is important to outline that  condition \fer{vic}, at least when $v_{m,\lambda}$ is equal to zero at the boundaries, is quite restrictive, and requires the initial data $v_0$  to be very close to the steady state. On the contrary,  if $\left(v_{m,\lambda}\right )^{-1}$ is bounded (and this happens when 
$\lim_{y \to -1^+} v_{m,\lambda}(y)= \lim_{y\to 1^-}v_{m,\lambda}(y)=+\infty$), condition \fer{vic} is satisfied any time $v_0$ is close to $v_{m,\lambda}$ in the $L^2$ distance.

In what follows, we will prove that exponential convergence  in $L^1$ can be obtained also for initial values more general than the ones satisfying Theorem \ref{l2}. To this extent, we will show that  the Beta functions \fer{beta}, in a certain well defined range of the parameters $\lambda$ and $m$, satisfy a weighted logarithmic-Sobolev inequality.   The result allows us to apply to our Fokker--Planck equation for opinion formation the same strategy one can apply to the classical Fokker--Planck equation \cite{AMTU}. 

Let us briefly recall the main steps of the (entropy) method for the classical one-dimensional Fokker--Planck equation.  Given the initial value problem
\be\label{FP}
\left\{
\begin{aligned}
&\partial_t f(t,x)= \partial_x^2 f(t,x) + \partial_x(xf(t,x)),\quad x\in \R, t>0\\
&f(0,x)=f_0(x) \ge 0 \in L^1(\R)
\end{aligned}
\right .
\ee
where the initial value is a probability density function, one studies the evolution of the relative entropy functional $H(f(t), M)$, given by
\be\label{entr-class}
H(f(t), M) = \int_\R f(t,x) \log \frac{f(t,x)}{M(x)} dx
\ee
where $M$ is the Maxwellian (Gaussian)
\be\label{Maxw}
M(x)= \frac 1{\sqrt{2\pi}} e^{-\frac {|x|^2}{2}},
\ee
which can be easily recognized as the unique  steady state of equation \fer{FP}. It is well known (cf. for example \cite{McK66}) that, if $f(t)$ is a solution of the Cauchy problem \eqref{FP},  the relative entropy is monotone nonincreasing, and its time derivative is given by
\be\label{derivata}
\frac d{dt} H(f(t), M) = - I(f(t), M), \quad t>0
\ee
where $I(f(t), M)$ is the relative Fisher information  (the entropy production) defined as
\be\label{fisher-class}
I(f(t), M) = \int_{\R} \left(
\partial_x \log  \frac {f(t,x)}{M(x)}\right )^2 f(t,x) dx.
\ee
Relation \eqref{derivata} coupled with the logarithmic-Sobolev inequality (cf. for example \cite{T99})
\[
H(f(t), M) \leq \frac 12 I (f(t),M),\quad t>0
\]
 leads to the exponential decay to zero of the relative entropy \cite{T99,T13}  with explicit rate.
 Last,  resorting to the well-known Csisz\'ar--Kullback--Pinsker inequality \cite{C} 
\be\label{CK}
\|f-g\|_{L^1}^2 \leq 2 H(f,g), \quad f,g\in L^1
\ee
one obtains exponential convergence in $L^1$ to the Maxwellian density (always with sub-optimal explicit rate). 

Going back to our problem, let us assume that the entropy of the initial value relative to the Beta steady state is bounded
\be\label{Hfinita}
H(v_0, v_{m,\lambda})  <\infty.
\ee
Evaluating the time derivative of the relative entropy  (cf. the computations in  \cite{FPTT17}), one obtains  for the solution to the Fokker--Planck equation \eqref{main} a relation analogous to \eqref{derivata}, which now reads
\be\label{derivata-peso}
\frac d{dt} H(v(t), v_{m,\lambda}) = - \tilde I(v(t), v_{m,\lambda}), \quad t>0.
\ee
In \fer{derivata-peso} $\tilde I$ defines the weighted Fisher information 
\be\label{w-fisher}
\tilde I(v(t), v_{m,\lambda}) = 
\int_\CI \frac \lambda 2 (1-y^2)
\left(\partial_y \log  \frac {v(t,y)}{ v_{m,\lambda}(y)}\right )^2 v(t,y) dy.
\ee
As one can easily verify, the weight $\frac \lambda 2 (1-y^2)$ is due to the  variable diffusion coefficient  in equation \eqref{main}.
It is clear that, if one can prove that, for some universal constant $C >0$ the relative entropy is bounded by 
\[
H(v, v_{m,\lambda}) \leq C \tilde I(v, v_{m,\lambda}),
\]
one obtains, as in the classical case, the exponential convergence to equilibrium of the relative entropy of the solution at the explicit rate $C$. 

We prove indeed that the following holds.
\begin{thm}\label{LS-theo}
Let $\lambda >0$, $m\in \CI$ be such that
\be\label{cond}
\begin{aligned}
& 1-\frac \lambda 2 >0, \quad m=0\\
& 1-\frac \lambda 2 \geq |m|, \quad m\neq 0.
\end{aligned}
\ee
and let $v_{m,\lambda}$ be the Beta function on $\CI$ defined by \fer{beta}. Then, there exists an explicit constant $K_{m,\lambda} >0$ such that,  for any
  probability density $\varphi \in L^1(\CI)$ absolutely continuous with respect to $v_{m,\lambda}$ it holds
  \be\label{LS-peso}
 H(\varphi, v_{m,\lambda}) \leq K_{m,\lambda} \tilde I(\varphi, v_{m,\lambda}).
 \ee
 The constant $K_{m,\lambda} >0$ is explicitly computable and equals
 \be\label{costante}
 K_{m,\lambda} =  \left(1-\frac \lambda 2 + \sqrt{\left(1-\frac \lambda 2\right )^2-m^2}\right)^{-1}.
  \ee
\end{thm}
\begin{rem}\label{remark}
It is worth underlying that conditions \eqref{cond} are equivalent to the condition that the corresponding Beta-type function $v_{m,\lambda}$ belongs to $L^2(\CI)$.
\end{rem}

A direct consequence of Theorem \ref{LS-theo} is the following
 \begin{thm}
 Let the parameters $\lambda >0$, $m\in \CI$ satisfy the conditions \fer{cond} of Theorem \ref{LS-theo}, and let  $v(t)=Q_tv_0$ be the solution to the initial-boundary value problem \fer{main} with no-flux boundary conditions, and initial data
 $v_0 \in L^1(\CI)$ a probability density such that the relative entropy $H(v_0, v_{m,\lambda})$ is finite. Then, the relative entropy decays exponentially to zero at an explicit rate, and
\be\label{conv}
\|v(t) -v_{m,\lambda}\|_{L^1} \leq\sqrt 2 e^{- \frac 1 {2K_{m,\lambda}} t} \sqrt {H(v_0, v_{m,\lambda})}.
\ee
In \fer{conv}  $K_{m,\lambda}>0$ is given by \fer{costante}.
 \end{thm}
\begin{proof}
We  already stressed in \eqref{ac} that, starting from the initial condition $v_0$,  the result by Epstein and Mazzeo  implies that the solution $v(t)=Q_t v_0$ defined in \eqref{Qt}  is absolutely continuous with respect to $v_{m,\lambda}$ for all $t>0$. Therefore, we can apply \eqref{derivata-peso} and then the weighted logarithmic-Sobolev inequality \eqref{LS-peso} with $\varphi(y)= v(t,y)$ for all $t>0$  to get
\[
\frac d{dt} H(v(t), v_{m,\lambda}) \leq - \frac 1{K_{m,\lambda} }H(v(t), v_{m,\lambda}), \quad t>0
\]
and this gives
\[
H(v(t), v_{m,\lambda}) \leq e^{-\frac 1{K_{m,\lambda} } t} H(v_0, v_{m,\lambda}).
\]
Then by the Csisz\'ar--Kullback--Pinsker inequality \eqref{CK}  we obtain
\be\label{conv-expL1-logsob}
\|v(t) -v_{m,\lambda}\|_{L^1} \leq \sqrt 2 e^{- \frac 1{2 K_{m,\lambda} } t} \sqrt {H(v_0, v_{m,\lambda})}, \quad t>0.
\ee
\end{proof}
Let us come back to the proof of Theorem \ref{LS-theo}.
The starting point is the well known Bakry--Emery result about logarithmic-Sobolev inequality.
\begin{thm}[Bakry--Emery \cite{BE}]
Let $M$ be a smooth, complete manifold and let $d\nu =e^{-\Psi}dx$ be a probability measure on $M$, such that $\Psi \in C^2(M)$ and $D^2\Psi + Ric \geq \rho I_n$, $\rho>0$.
Then, for every probability measure $\mu$ absolutely continuous with respect to $\nu$, we have
\be\label{BE}
H(\mu,\nu) \leq \frac 1{2\rho} I(\mu,\nu)
\ee
where
\[
H(\mu,\nu)=\int_M \log \frac {d \mu}{d \nu} d\mu 
\]
and
\[
I(\mu, \nu) = \int_M \left|\nabla \log \frac{d \mu}{d \nu} \right |^2 d \mu .
\]
\end{thm}
If $M=[a,b]$ is an interval of the real line, $d \nu = g dx$ and $ d\mu= f dx$, with $f$ and $g$ probability densities, the assumptions in Bakry--Emery criterion read as follows
\be
\begin{aligned}
&g(x)=e^{-\Psi(x)},\\
&\Psi \in C^2([a,b])\\
&\min_{[a,b]}\Psi''(x)\geq \rho >0.
\end{aligned} 
\ee
Then, for any $f$ probability density on $[a,b]$ absolutely continuous with respect to $g$, inequality \eqref{BE} becomes
\be\label{BE-int}
\int_a^b f(x) \log \frac {f(x)}{g(x)} dx \leq \frac 1{2\rho} \int_a^b  \left(\frac d{dx} \log \frac{f(x)}{g(x)} \right )^2f(x) dx.
\ee
Of course this is a non--weighted logarithmic-Sobolev result. We are going to identify who will play the role of $\mu$ and $\nu$. If we take $M=\CI$ and $\nu= v_{m,\lambda} dx$ then two problems appear.
The first one is that the open interval $\CI$ is not a complete manifold and the other one is that even if we prove that $v_{m,\lambda}(y) = e^{- \Psi (y)}$ with $\Psi$ satisfying Bakry--Emery Theorem, then for any $\varphi$ probability density absolutely continuous with respect to 
$v_{m,\lambda}$ we would get the logarithmic-Sobolev inequality
\[
\int_{-1}^1 \varphi(x) \log \frac {\varphi(x)}{v_{m,\lambda}(x)} dx \leq \frac 1{2\rho} \int_{-1}^1  \left( \frac d{dx} \log \frac {\varphi(x)}{v_{m,\lambda}(x)}\right )^2 \varphi(x) dx.
\]
This is not enough to obtain \eqref{LS-peso} since $\frac \lambda 2 (1-y^2) \leq 1$ for $\lambda <2$ (which is implied by conditions \eqref{cond}).
It turns out that actually  $v_{m,\lambda}$ satisfies  $v_{m,\lambda}(y) = e^{- \Psi (y)}$ with $\Psi$ fulfilling Bakry--Emery conditions.  Since we are going to prove a stronger inequality in a different way, we leave the details to the interested reader.

\medskip
\noindent{\bf Proof of Theorem \ref{LS-theo}.}
The main idea is to resort to a change of variable which transforms the weighted logarithmic-Sobolev inequality \eqref{LS-peso} we are looking for into a usual logarithmic-Sobolev inequality for a different probability density which satisfies the assumptions of the Bakry--Emery criterion.
Given the partial differential equation
\[
\partial_t v(t,y)= \frac \lambda 2 \partial_y^2 \left((1-y^2) v(t,y)\right ) +\partial_y \left((y-m)v(t,y)\right ),\quad  t>0,\quad y \in \CI
\]
with steady state $v_{m,\lambda}$, its adjoint equation reads
\be\label{main-adj}
\partial_t u(t,x)= \frac \lambda 2 (1-x^2)  \partial_x^2 u(t,x) - (x-m)\partial_x u(t,x), \quad  t>0,\quad x \in \CI.
\ee
If we now set in \fer{main-adj}
\[
f(t,s)=u(t,x)
\]
where
\[
\frac {ds}{dx} = \frac 1{\sqrt {1-x^2}}, \quad x\in \CI,
\]
 equation \eqref{main-adj} is transformed into a Fokker--Planck equation with constant diffusion, given by
\be\label{FP-adj-free}
\partial_t f(t,s)= \frac \lambda 2   \partial_s^2 f(t,s) - \frac {\left(1-\frac \lambda 2\right ) \sin s -m}{\cos s}\, \partial_s f(t,s), \quad  t>0, s\in \left(-\frac \pi 2, \frac \pi 2\right ).
\ee
The adjoint equation of \eqref{FP-adj-free} is in turn
\be\label{FP-free}
\partial_t g(t,z)= \frac \lambda 2   \partial_z^2 g(t,z) + \partial_z \left( \frac {\left(1-\frac \lambda 2\right ) \sin z -m}{\cos z}\,  g(t,z)\right ), \quad  t>0, z\in \left(-\frac \pi 2, \frac \pi 2\right ).
\ee
We denote
\be\label{w'}
W_{m,\lambda}'(z):=  \frac {\left(1-\frac \lambda 2\right ) \sin z -m}{\cos z}
\ee
and
\[
W_{m,\lambda}(z)=  \int _{0}^z W_{m,\lambda}'(\sigma) d \sigma.
\]
The steady states of  Equation \eqref{FP-free} are
\be\label{exp}
g_{m,\lambda}(z)= C_{m,\lambda} e^{-\frac 2  \lambda W_{m,\lambda}(z)} = e^{-\left(\frac 2  \lambda W_{m,\lambda}(z) -\log P_{m,\lambda} \right )}
\ee
for $C_{m,\lambda} >0$ as in \eqref{beta}
and explicitly
\be\label{beta-2}
g_{m,\lambda}(z)=  C_{m,\lambda} \frac 1{(\cos z)^{ 1-\frac  2\lambda}} \frac {\left(1+\tan \frac z2\right )^{\frac {2m}{\lambda}}}{\left(1-\tan \frac z2\right )^{\frac {2m}{\lambda}}}.
\ee
One can check that
\[
\begin{aligned}
& g_{m,\lambda}(z) \sim R_{m,\lambda}\left( \frac \pi 2-z\right )^{\frac 2\lambda -1-\frac {2m}\lambda},\quad z\to \frac \pi 2^-\\
& g_{m,\lambda}(z) \sim \tilde R_{m,\lambda}\left( \frac \pi 2+z\right )^{\frac 2\lambda -1+\frac {2m}\lambda},\quad z\to -\frac \pi 2^+
\end{aligned}
\]
with $R_{m,\lambda}$, $\tilde R_{m,\lambda}$ positive constants.
Moreover, we have
\be\label{stati-staz}
\frac {g_{m,\lambda}(\arcsin y)}{\sqrt{1-y^2}} =  v_{m,\lambda}(y),\quad y\in \CI
\ee
with $v_{m,\lambda}$ as in \eqref{beta} or, equivalently,
\[
g_{m,\lambda} (z)=  v_{m,\lambda}(\sin z) \cos z, \quad z\in \left( -\frac \pi 2, \frac \pi 2\right ).
\]
It is immediate to show that $g_{m,\lambda}$ satisfies the assumptions of Bakry--Emery criterion on $ \left(-\frac \pi 2, \frac \pi 2\right )$. Since the latter is an open interval (and so it is not a complete manifold), we will overcome this difficulty by a suitable approximation argument.
Resorting to \eqref{exp}, we need to evaluate $\frac 2\lambda W_{m,\lambda}''(z)$. We obtain
\[
\frac 2\lambda W_{m,\lambda}''(z)= \frac 2\lambda  \frac d{dz} W_{m,\lambda}'(z) =\frac 2\lambda  \frac{\left(1-\frac \lambda 2\right )+ m\sin z}{\cos^2 z}, \quad z\in \left( -\frac \pi 2, \frac \pi 2\right ).
\]
Therefore, provided $1-\frac \lambda 2 \geq |m|$, 
\[
\inf_{ \left( -\frac \pi 2, \frac \pi 2\right )} W_{m,\lambda}''(z) \ge 0
\]
and the function $W_{m,\lambda}(z)$ is convex on $ \left( -\frac \pi 2, \frac \pi 2\right )$.
If $m=0$,  
\be\label{rho_0}
\min_{\left( -\frac \pi 2, \frac \pi 2\right )} W_{0,\lambda}''(z)= W_{0,\lambda}''(0)=  1-\frac\lambda 2:= \rho_{0,\lambda}.
\ee
Consequently, in order to apply Bakry--Emery criterion, we have to assume $1-\frac \lambda 2 >0$. 

Let us now set $m \not=0$.
Since
\[
\frac d {dz} W_{m,\lambda}''(z)= \left(-\frac 1{\cos^3 z}\right )\left(m\sin^2 z +(\lambda-2)\sin z+m\right ),
\]
for any given $m\in \CI$, $m\neq 0$ and $\lambda$ such that $1-\frac \lambda 2 \geq |m|$,  there exists $\bar z_{m,\lambda} \in \left(-\frac \pi 2,\frac \pi 2\right )$ such that 
\be\label{rho_mlambda}
\min_{\left( -\frac \pi 2, \frac \pi 2\right )} W_{m,\lambda}''(z)= W_{m,\lambda}''(\bar z_{m,\lambda})= 
\frac 12 \left(1-\frac \lambda 2 + \sqrt{\left(1-\frac \lambda 2\right )^2-m^2}\right) := \rho_{m,\lambda}>0.
\ee
If we could apply Bakry--Emery criterion directly on $ \left(-\frac \pi 2,\frac \pi 2\right )$ we would obtain, for all $f$ probability densities on $ \left(-\frac \pi 2,\frac \pi 2\right )$ absolutely continuous with respect to $g_{m,\lambda}$, the logarithmic Sobolev inequality
\[
\int_{-\frac \pi 2}^{\frac \pi 2} f(z) \log \frac {f(z)}{g_{m,\lambda}(z)} dz \leq \frac \lambda{4 \rho_{m,\lambda}} \int_{-\frac \pi2}^{\frac \pi 2} \left(\frac d{dz} \log \frac{f(z)}{g_{m,\lambda}(z)} \right )^2f(z) dz, 
\] 
where the explicit constants $\rho_{m,\lambda}$ are defined  in \eqref{rho_0} and \eqref{rho_mlambda}.
Since $\left(-\frac \pi 2,\frac \pi 2\right )$ is not a complete manifold we perform an approximation argument.
Let us fix $m\in \CI$ and $\lambda >0$ satisfying \eqref{cond} and let $f$ be a probability density on $\left(-\frac \pi 2,\frac \pi 2\right )$ absolutely continuous with respect to $g_{m,\lambda}$.
For $\eps >0$ let us define
\begin{align*}
&
f_\eps=\frac 1{A_\eps} f\chi_{\left[-\frac \pi 2+\eps,\frac \pi 2-\eps \right]}, \text{ with } A_\eps = \int_{-\frac \pi 2 +\eps} ^{\frac \pi 2 -\eps} f(z) dz\\
&
g_\eps = \frac 1{B_\eps}g_{m,\lambda}\chi_{\left[-\frac \pi 2+\eps,\frac \pi 2-\eps \right]}, \text{ with } B_\eps = \int_{-\frac \pi 2 +\eps} ^{\frac \pi 2 -\eps} g_{m,\lambda}(z) dz.
\end{align*}
Of course $f_\eps$ and $g_\eps$ are probability densities and $A_\eps \to 1$, $B_\eps \to 1$ for $\eps \to 0$. Moreover by \eqref{exp}
\[
g_\eps (z)= e^{-\left(\frac 2  \lambda W_{m,\lambda}(z) -\log P_{m,\lambda} + \log B_\eps \right )} \chi_{\left[-\frac \pi 2+\eps,\frac \pi 2-\eps \right]}(z)
\]
and $f_\eps$ is absolutely continuous with respect to $g_\eps$ on $\left[-\frac \pi 2+\eps,\frac \pi 2-\eps \right]$. 
For all $\eps >0$ we have
\[
\frac {d^2}{dz^2} \left(\frac 2  \lambda W_{m,\lambda}(z) -\log P_{m,\lambda} + \log B_\eps\right ) = \frac 2  \lambda W''_{m,\lambda}(z) \geq \frac 2\lambda \rho_{m,\lambda}.
\]
Since $g_\eps$ satisfies the assumptions of Bakry--Emery criterion on $\left[-\frac \pi 2+\eps,\frac \pi 2-\eps \right]$, we get for all $\eps >0$
\be\label{BE-approx}
\int_{-\frac \pi 2+\eps}^{\frac \pi 2-\eps} f_\eps(z) \log \frac {f_\eps(z)}{g_\eps(z)} dz \leq \frac \lambda{4 \rho_{m,\lambda}} \int_{-\frac \pi2+\eps}^{\frac \pi 2-\eps} \left(\frac d{dz} \log \frac{f_\eps(z)}{g_\eps(z)} \right )^2f_\eps(z) dz.
\ee
Now assume that
\be\label{Fisher-bounded}
 \int_{-\frac \pi2}^{\frac \pi 2} \left(\frac d{dz} \log \frac{f(z)}{g_{m,\lambda}(z)} \right )^2f(z) dz <\infty.
\ee
As far as the right hand side of \eqref{BE-approx} is concerned, by Lebesgue's dominated convergence theorem we get for $\eps \to 0$
\[
\begin{aligned}
&\int_{-\frac \pi2+\eps}^{\frac \pi 2-\eps} \left(\frac d{dz} \log \frac{f_\eps(z)}{g_\eps(z)} \right )^2f_\eps(z) dz\\
& = \frac 1 {A_\eps} \int_{-\frac \pi 2}^{\frac \pi 2}   \left(\frac d{dz} \log  \left(\frac {f(z)}{A_\eps} \frac {B_\eps}{g_{m,\lambda}(z)} \right ) \right )^2   f(z)  \chi_{\left[-\frac \pi 2+\eps,\frac \pi 2-\eps \right]}dz\\
& = \frac 1 {A_\eps} \int_{-\frac \pi 2}^{\frac \pi 2}   \left(\frac d{dz} \left(\log  \frac{f(z)}{g_{m,\lambda}(z)} + \log \frac {B_\eps}{A_\eps}\right )  \right )^2   f(z)  \chi_{\left[-\frac \pi 2+\eps,\frac \pi 2-\eps \right]}(z) dz\\
& = \frac 1 {A_\eps} \int_{-\frac \pi 2}^{\frac \pi 2}   \left(\frac d{dz} \log  \frac{f(z)}{g_{m,\lambda}(z)}   \right )^2   f(z)  \chi_{\left[-\frac \pi 2+\eps,\frac \pi 2-\eps \right]}(z) dz \to 
 \int_{-\frac \pi2}^{\frac \pi 2} \left(\frac d{dz} \log \frac{f(z)}{g_{m,\lambda}(z)}\right )^2f(z) dz.
\end{aligned}
\]
Letting $\eps\to 0$, for the left hand side we obtain
\[
\begin{aligned}
&\int_{-\frac \pi 2+\eps}^{\frac \pi 2-\eps} f_\eps(z) \log \frac {f_\eps(z)}{g_\eps(z)} dz =  \int_{-\frac \pi 2}^{\frac \pi 2} \frac {f(z)}{A_\eps}  \log  \left(\frac {f(z)}{A_\eps} \frac {B_\eps}{g_{m,\lambda}(z)} \right ) \chi_{\left[-\frac \pi 2+\eps,\frac \pi 2-\eps \right]}(z)dz\\
& =  \frac 1{A_\eps} \int_{-\frac \pi 2}^{\frac \pi 2}f(z)  \log  \frac {f(z)}{g_{m,\lambda}(z)}  \chi_{\left[-\frac \pi 2+\eps,\frac \pi 2-\eps \right]}(z) dz + 
\frac 1{A_\eps}   \log \frac  {B_\eps}  {A_\eps}    \int_{-\frac \pi 2}^{\frac \pi 2} f(z) \chi_{\left[-\frac \pi 2+\eps,\frac \pi 2-\eps \right]}(z) dz \\
&\to   \int_{-\frac \pi 2}^{\frac \pi 2} f(z) \log \frac {f(z)}{g_{m,\lambda}(z)} dz.
\end{aligned}
\]
Indeed,  by Lebesgue's dominated convergence theorem
\[
\frac 1{A_\eps}   \log \frac  {B_\eps}  {A_\eps}    \int_{-\frac \pi 2}^{\frac \pi 2} f(z) \chi_{\left[-\frac \pi 2+\eps,\frac \pi 2-\eps \right]}(z) dz \to  0, \quad \eps \to 0,
\]
and thanks to the identity
\[
\begin{aligned}
& \int_{-\frac \pi 2}^{\frac \pi 2}f(z)  \log  \frac {f(z)}{g_{m,\lambda}(z)}  \chi_{\left[-\frac \pi 2+\eps,\frac \pi 2-\eps \right]}(z) dz\\
&=  \int_{-\frac \pi 2}^{\frac \pi 2}\left( \frac{f(z)}{ g_{m,\lambda}(z)} \log  \frac {f(z)}{g_{m,\lambda}(z)}  -\frac{f(z)}{ g_{m,\lambda}(z)} +1\right ) g_{m,\lambda}(z) \chi_{\left[-\frac \pi 2+\eps,\frac \pi 2-\eps \right]}(z) dz\\
&\quad\quad + \int _{-\frac \pi 2}^{\frac \pi 2}\left(f(z)- g_{m,\lambda}(z)\right )  \chi_{\left[-\frac \pi 2+\eps,\frac \pi 2-\eps \right]}(z) dz,
\end{aligned}
\]
by the Lebesgue's dominated   and monotone convergence theorems we conclude
\[
 \frac 1{A_\eps} \int_{-\frac \pi 2}^{\frac \pi 2}f(z)  \log  \frac {f(z)}{g_{m,\lambda}(z)}  \chi_{\left[-\frac \pi 2+\eps,\frac \pi 2-\eps \right]}(z) dz \to   \int_{-\frac \pi 2}^{\frac \pi 2}f(z)  \log  \frac {f(z)}{g_{m,\lambda}(z)} dz, \quad \eps \to 0.
 \]
Finally, for all $f$ probability densities on $ \left(-\frac \pi 2,\frac \pi 2\right )$ absolutely continuous with respect to $g_{m,\lambda}$ it holds
\be\label{logg}
\int_{-\frac \pi 2}^{\frac \pi 2} f(z) \log \frac {f(z)}{g_{m,\lambda}(z)} dz \leq \frac \lambda {4 \rho_{m,\lambda}} \int_{-\frac \pi2}^{\frac \pi 2} \left(\frac d{dz} \log \frac{f(z)}{g_{m,\lambda}(z)} \right )^2f(z) dz,
\ee
where $\rho_{m,\lambda}$ are defined as in \eqref{rho_0} and \eqref{rho_mlambda}.
Going back to the original functions, by means of the change of variables 
\[
z=\arcsin y
\]
 the logarithmic Sobolev inequality \fer{logg} transforms into a weighted logarithmic-Sobolev inequality. In fact,   for any $f$ probability density on $\left(-\frac \pi 2,\frac \pi 2\right )$ absolutely continuous with respect to $g_{m,\lambda}$ 
\begin{multline}
\int_{-1}^{1} f(\arcsin y) \log \frac {f(\arcsin y)}{g_{m,\lambda}(\arcsin y)} \frac 1{\sqrt{1-y^2}} dy  \leq \\
\frac 1{2\rho_{m,\lambda}} \int_{-1}^{1}  \frac \lambda 2 \left(\frac d{dy} \log \left(\frac{f(\arcsin y)}{g_{m,\lambda}(\arcsin y)}\right ) \sqrt{1-y^2} \right )^2f(\arcsin y) \frac 1{\sqrt{1-y^2}} dy.
\end{multline}
Now, by \eqref{stati-staz} we get
\begin{multline}
\int_{-1}^{1}  \frac {f(\arcsin y) } {\sqrt{1-y^2}}  \log \frac { \frac {f(\arcsin y) } {\sqrt{1-y^2}} }{v_{m,\lambda}(y)} dy  \leq \\
\frac 1{2\rho_{m,\lambda}} \int_{-1}^{1}  \frac \lambda 2(1-y^2) \left(\frac d{dy} \log \left(\frac{ \frac {f(\arcsin y) } {\sqrt{1-y^2}} }{v_{m,\lambda}(y)}\right ) \right )^2  \frac {f(\arcsin y) }{\sqrt{1-y^2}} dy.
\end{multline}
In order to complete the proof of inequality \eqref{LS-peso} it enough to observe that $\varphi \in L^1(\CI)$ is a probability density absolutely continuous with respect to $v_{m,\lambda}$ if and only if $\varphi(y)=\frac {f(\arcsin y) } {\sqrt{1-y^2}}$
with $f \in L^1 \left(\left(-\frac \pi 2,\frac \pi 2\right )\right )$ is a probability density absolutely continuous with respect to $g_{m,\lambda}$. 
Inequality \eqref{LS-peso} is then proven with 
\be\label{Kmlambda}
K_{m,\lambda} = \frac 1{2\rho_{m,\lambda}}.
\ee

\medskip
\hfill$\square$

\begin{rem} It is worth comparing the results of exponential convergence in $L^1$ contained in \eqref{conv-expL1-L2} and \eqref{conv-expL1-logsob}.
By \eqref{Kmlambda}, for $m$ and $\lambda$ satisfying conditions \eqref{cond}  we get  \eqref{conv-expL1-logsob}:
\[
\|v(t) -v_{m,\lambda}\|_{L^1} \leq \sqrt 2 e^{- \rho_{m,\lambda}  t} \sqrt {H(v_0, v_{m,\lambda})}, \quad t>0,
\]
with $\rho_{m,\lambda}$ as in \eqref{rho_mlambda}.
On the other hand, for all $m\in \CI$ and $\lambda >0$ we get \eqref{conv-expL1-L2} :
\[
\left \| v(t)-v_{m,\lambda}\right \|_{L^1} \leq e^{-t} \left( \int_{-1}^1 \frac {(v_0(y)-v_{m,\lambda}(y))^2}{v_{m,\lambda}(y)} \, dy\right )^{\frac 12},\quad t>0.
\]
Since $\rho_{m,\lambda}\leq 1$
for all $m$, $\lambda$ satisfying conditions \eqref{cond}, the rate of exponential convergence in the second estimate is sharper than the first one.
Let us compare now the assumptions
\[
H(v_0, v_{m,\lambda}) = \int_{-1}^1 v_0(y)\log \frac {v_0(y)}{v_{m,\lambda}(y)}\, d y< \infty
\]
and
\[
\int_{-1}^1 \frac {(v_0(y)-v_{m,\lambda}(y))^2}{v_{m,\lambda}(y)} dy <\infty
\]
for the values of the parameters which fulfill conditions \eqref{cond}.
Since 
\[
x\log x  \geq x-1 + \frac 12 (x-1)^2\chi_{\left\{x\leq 1\right\} }(x), \quad x>0
\]
we get
\[
\begin{aligned}
&\int_{-1}^1 v_{m,\lambda}(y) \frac {v_0(y)}{v_{m,\lambda}(y)}\log \frac {v_0(y)}{v_{m,\lambda}(y)}\, d y\\
& \geq \int_{-1}^1 v_{m,\lambda}(y)\left( \frac {v_0(y)}{v_{m,\lambda}(y)} -1\right )\, d y + \frac 12 \int_{-1}^1 v_{m,\lambda}(y)\left( \frac {v_0(y)}{v_{m,\lambda}(y)} -1\right )^2 \chi_{\left\{v_0(y) \leq v_{m,\lambda}(y)    \right \}}(y) \, d y\\
&= \frac 12 \int_{-1}^1  \frac {(v_0(y)-v_{m,\lambda}(y))^2}{v_{m,\lambda}(y)} \chi_{\left\{v_0(y) \leq v_{m,\lambda}(y)    \right \}}(y) \, d y.
\end{aligned}
\]
So for $v_0 \leq v_{m,\lambda}$ the rate of convergence contained in \eqref{conv-expL1-L2} is stronger than that in 
 \eqref{conv-expL1-logsob}.
Moreover, 
\[
x\log x  \leq x-1 + \frac 12 (x-1)^2, \quad x\geq 1
\]
and so for $v_0 \geq v_{m,\lambda}$ we get
\[
\frac 12 \int_{-1}^1  \frac {(v_0(y)-v_{m,\lambda}(y))^2}{v_{m,\lambda}(y)}  \, d y \geq \int_{-1}^1 v_{m,\lambda}(y) \frac {v_0(y)}{v_{m,\lambda}(y)}\log \frac {v_0(y)}{v_{m,\lambda}(y)}\, d y.
\]
In this case, the convergence obtained by the new weighted logarithmic-Sobolev inequality could be the only one available.
Of course, in all the other cases the two conditions seem not to be comparable.
\end{rem}

\section{A distinguished case}\label{dist}

From Theorem \ref{LS-theo} one can extract some interesting consequences. The case $m =0$,  $\lambda = 1$ corresponds to the uniform density
 \[
 v_{0,1}(x) = \frac 12, \quad x \in \CI.
 \] 
Hence, considering that $K_{0,1} = 1$, for a given probability density $h$ on $\CI$, inequality \fer{LS-peso} takes the form
 \be\label{poi}
 \int_\CI  h(x) \log h(x) \, dx + \log 2 \le \frac 12 \int_\CI (1-x^2) \frac{(h'(x))^2}{h(x)}\, dx.
 \ee
A more suitable form is obtained by setting $h(x) = f^2(x)$ into \fer{poi}. One obtains the inequality
\be\label{LS-d}
 \int_\CI  f^2(x) \log f^2(x) \, dx + \log 2 \le 2 \int_\CI (1-x^2) (f'(x))^2\, dx,
 \ee
satisfied by all functions $f$ in $L^2(\CI)$ of $L^2$-norm equal to one.  Inequality \fer{LS-d} is the analogous of the standard \emph{Euclidean logarithmic-Sobolev inequality} established in Gross \cite{Gro}, which in  one-dimension reads
\be\label{LSI}
 \int_\R f^2(x) \log f^2(x) \, dx + \frac 12 \log(2\pi e^2) \le  2 \int_\R (f'(x))^2\, dx,
 \ee
and it is valid for all functions $f$ such that
 \[
 \int_\R f(x)^2\, dx = \int_\R x^2f^2(x) \, dx = 1.
 \]
Note that the main difference between the logarithmic-Sobolev inequality \fer{LSI} and the new inequality \fer{LS-d}, apart from the different interval of integration, is the presence of the weight on the right-hand side. 

Clearly, the constraint $\|f\|_2 = 1$ can be easily cut to give the (general) inequality
\be\label{poi1}
 \int_\CI  w^2(x) \log w^2(x) \, dx - \|w\|_2^2 \log \frac{\|w\|_2^2}2 \le  2 \int_\CI (1-x^2) (w'(x))^2\, dx,
\ee
which is valid for any function $w \in L^2(\CI)$.

\section{Numerical experiments}\label{nume}

In this short Section, we will focus on some numerical experiments that illustrate the time-evolution of the weighted logarithmic Sobolev inequality \fer{LS-peso} for various values of the parameter $\lambda$, and $m =0$. To this extent, we make use of numerical schemes for the Fokker--Planck equation \fer{op-FP}, recently considered in \cite{PZ}, that preserve the structural properties, like non negativity of the solution, entropy dissipation and large time behavior.  These properties are essential for a correct description of the underlying physical problem. 

The experiments have been done by choosing as initial density a bimodal normal distribution centered in $\pm 1/2$, normalized in the interval $(-1,1)$. 
It is clearly shown in Figure \fer{fig:test1} that inequality \fer{LS-peso} gives a better approximation to the entropy decay towards equilibrium for small values of the parameter. In all cases, however, exponential in time decay follows. 

\begin{figure}\centering
    {\includegraphics[width=7cm]{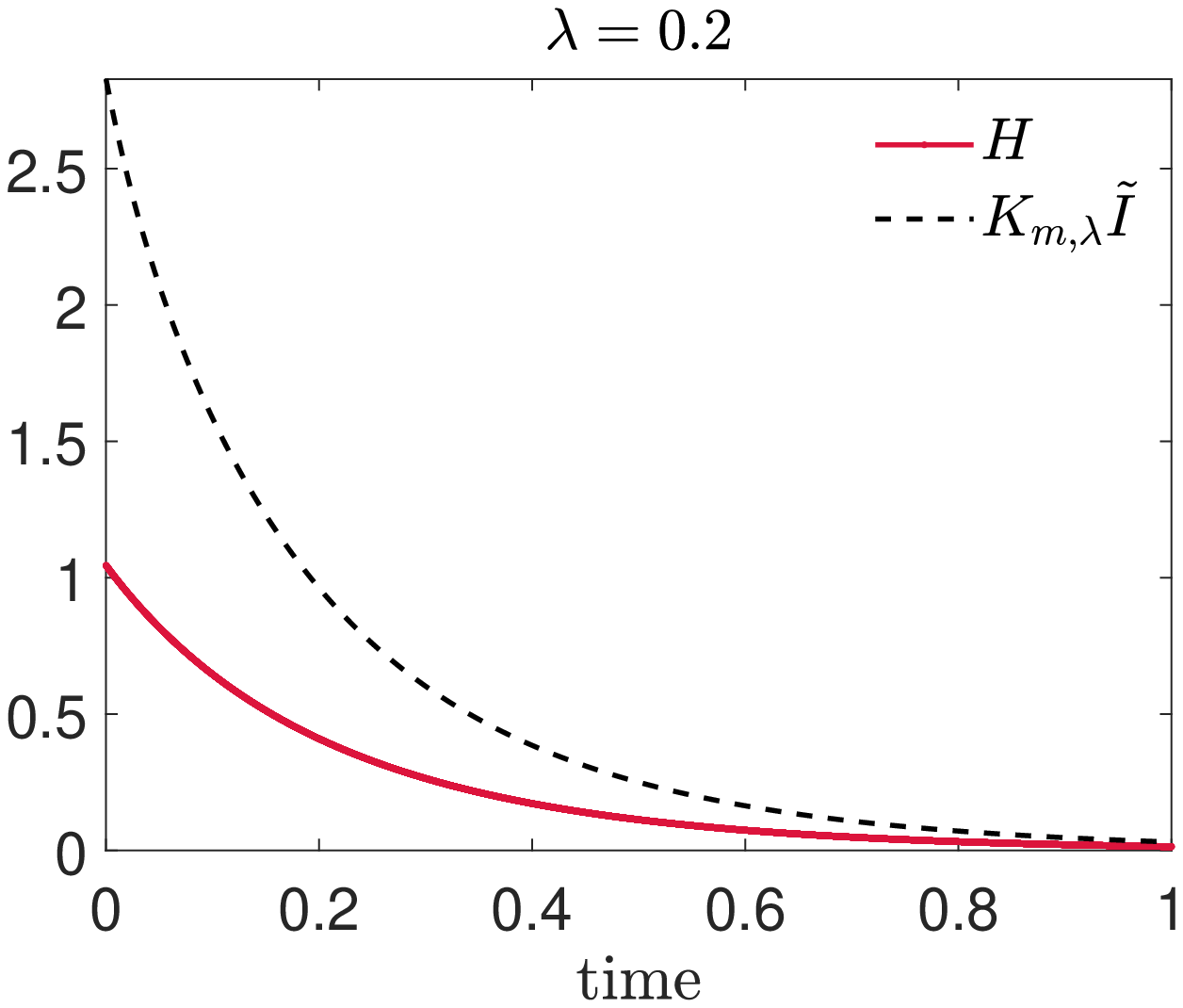}}
    \hspace{+0.35cm}
    {\includegraphics[width=7cm]{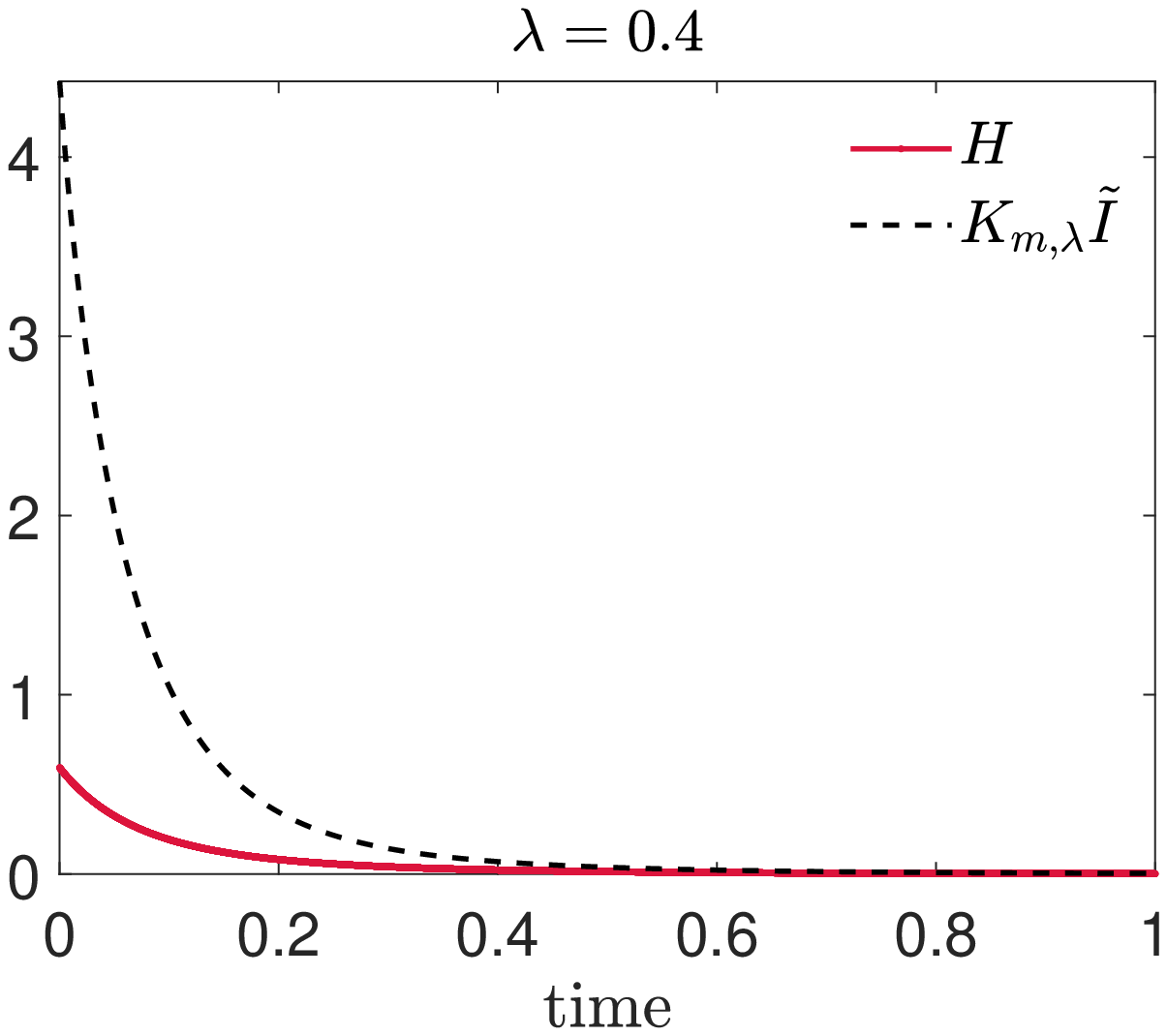}}\\
    \vspace{+0.45cm}
    {\includegraphics[width=7cm]{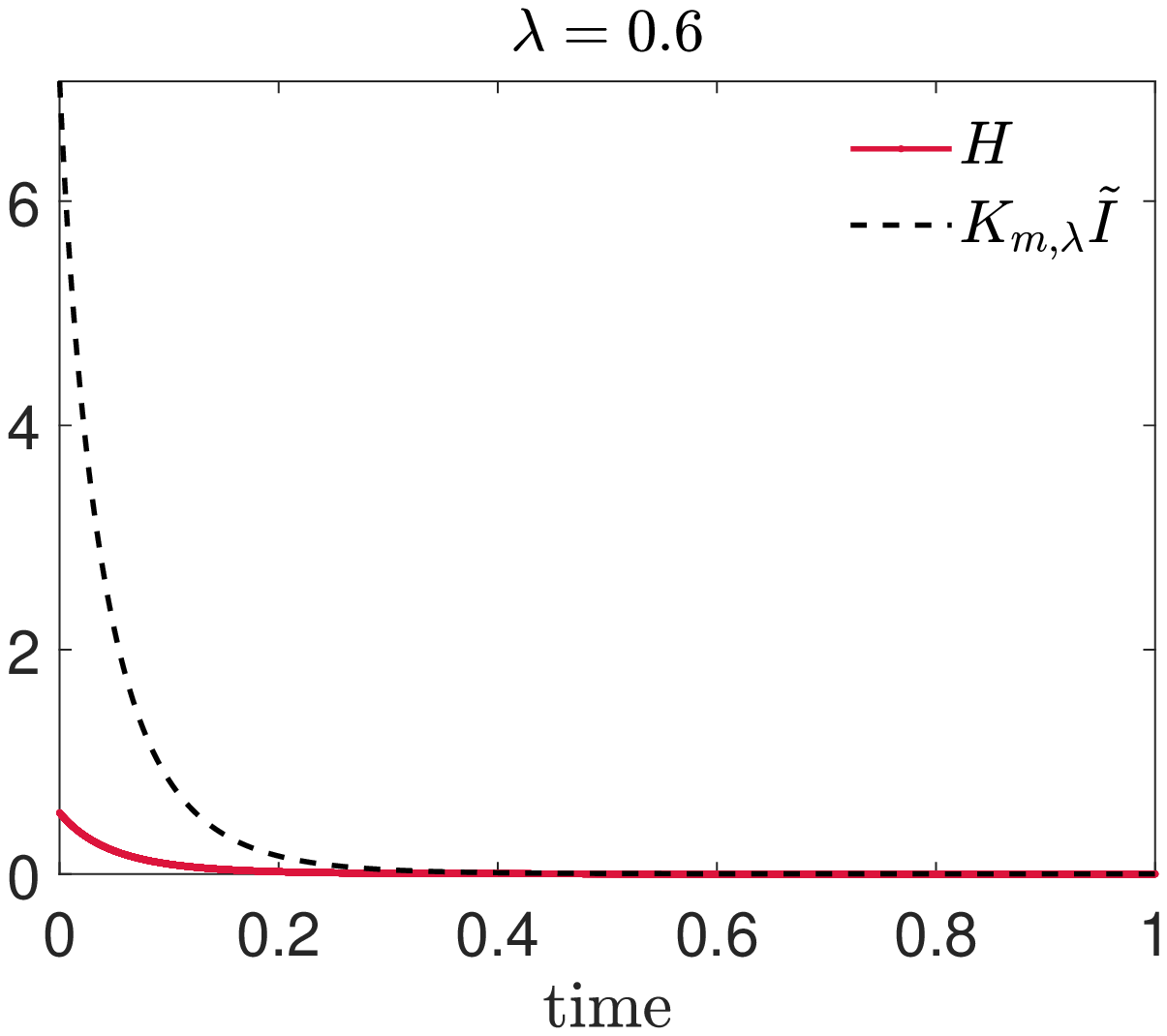}}
    \vspace{+0.35cm}
    {\includegraphics[width=7cm]{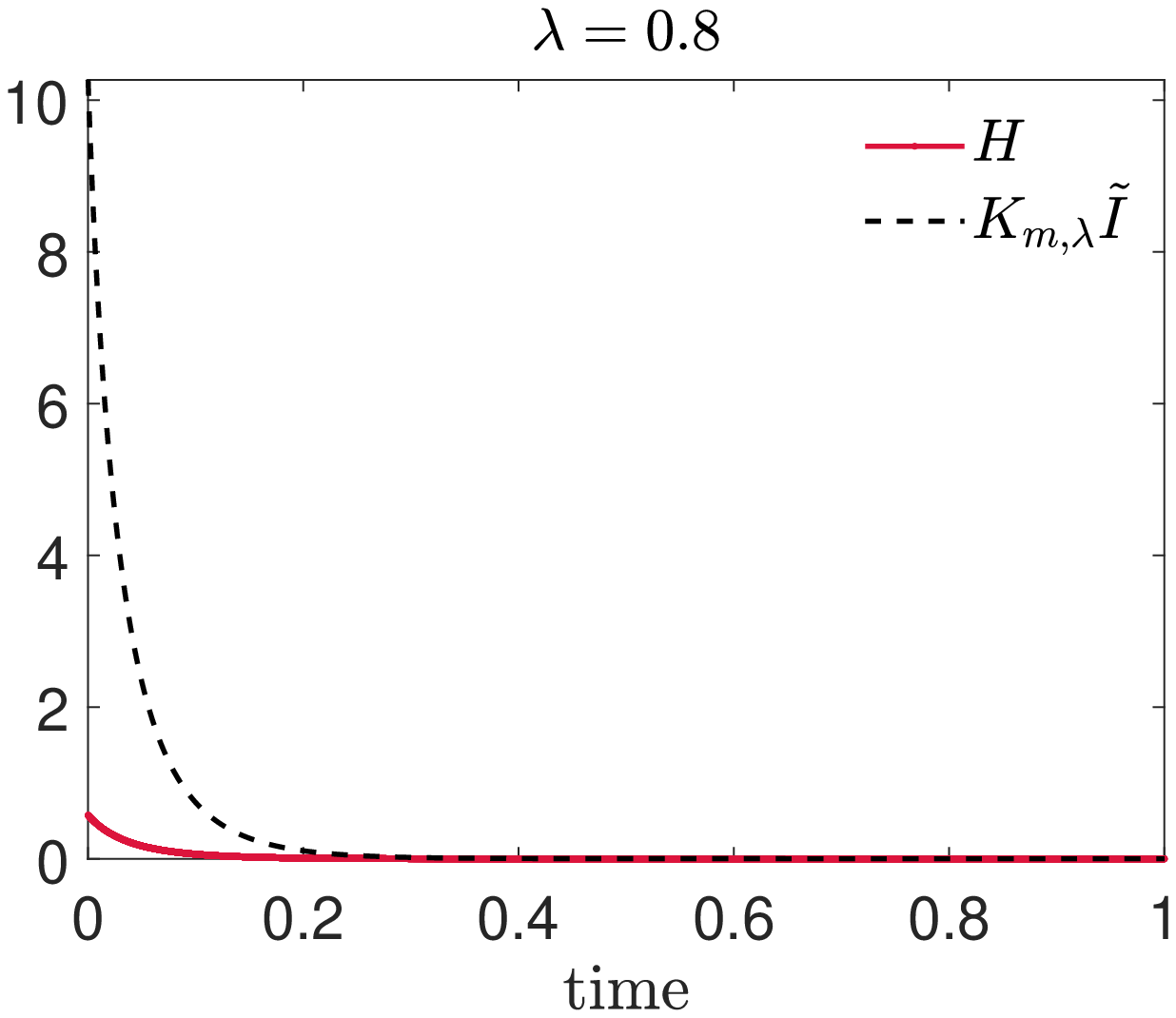}}
    \caption{Time evolution of the weighted logarithmic Sobolev inequality \fer{LS-peso} for the Fokker--Planck model as a function of the  parameter $\lambda$.  }\label{fig:test1}
\end{figure}

In Figure \fer{fig:test2} it is shown that the numerical method correctly reproduce the equilibrium Beta density \fer{beta} of the Fokker--Planck equation \fer{op-FP} for any value of the parameter $\lambda$. 
\begin{figure}\centering
    {\includegraphics[width=7cm]{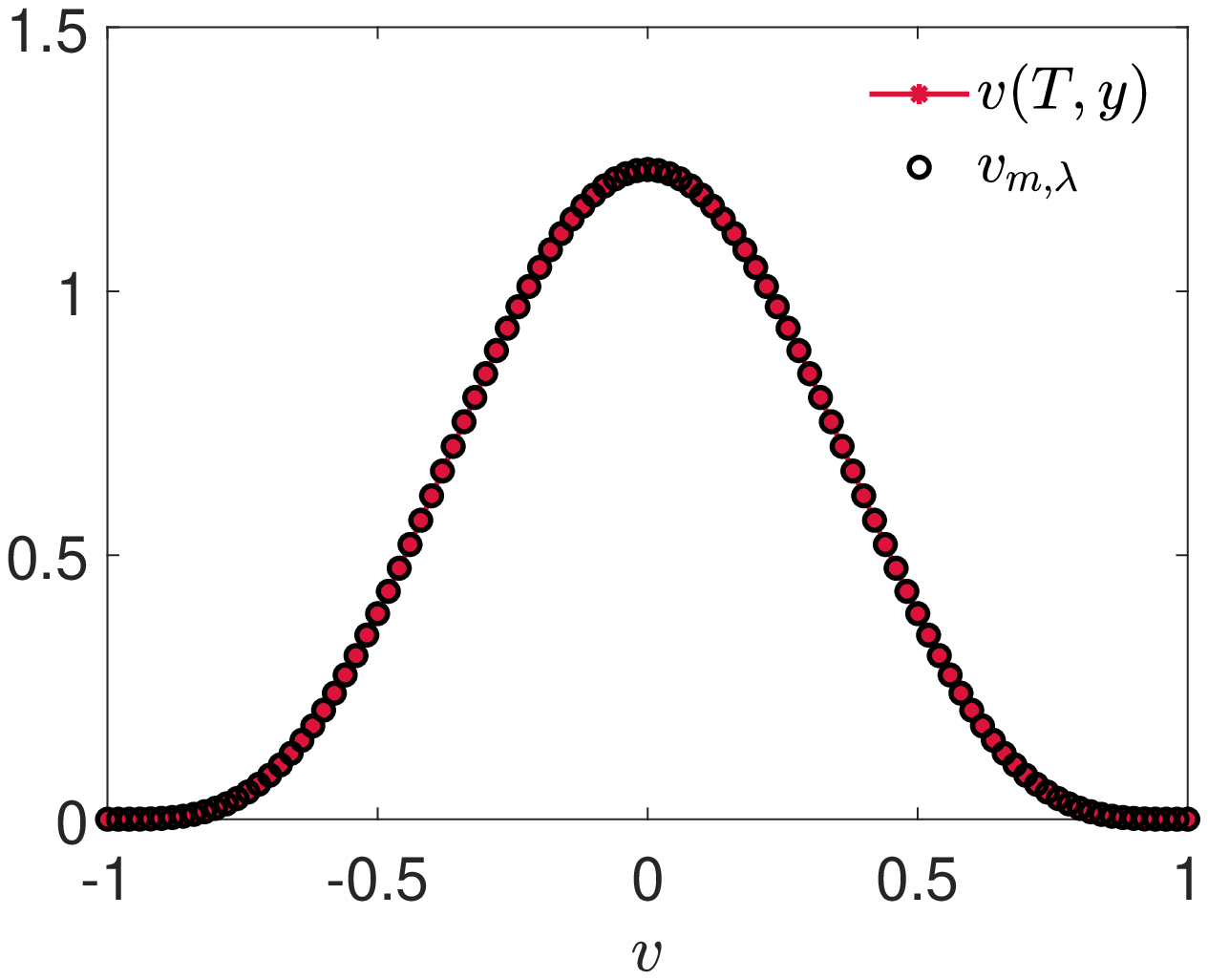}}
    \hspace{+0.35cm}
    {\includegraphics[width=7cm]{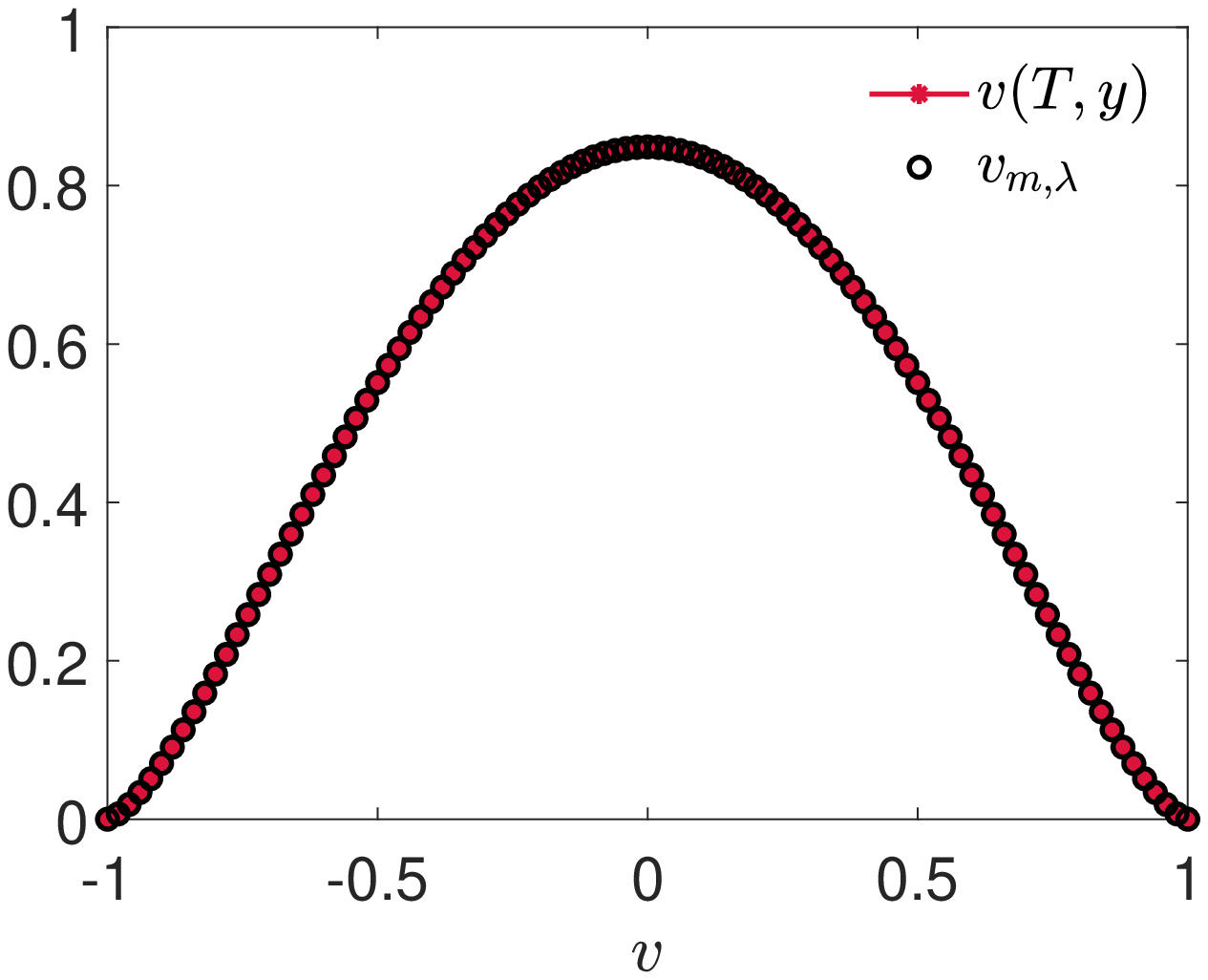}}\\
    \vspace{+0.45cm}
    {\includegraphics[width=7cm]{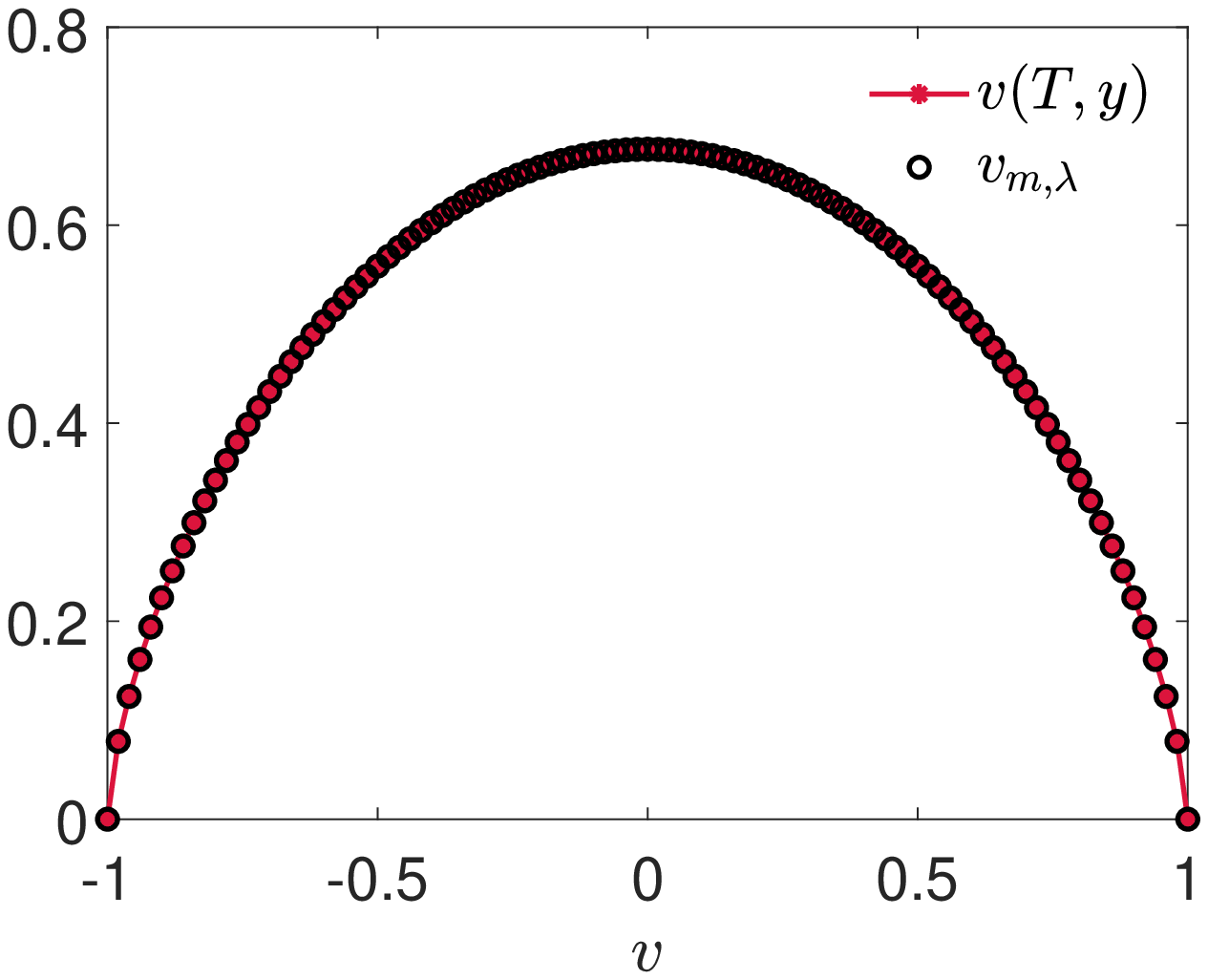}}
    \vspace{+0.35cm}
    {\includegraphics[width=7cm]{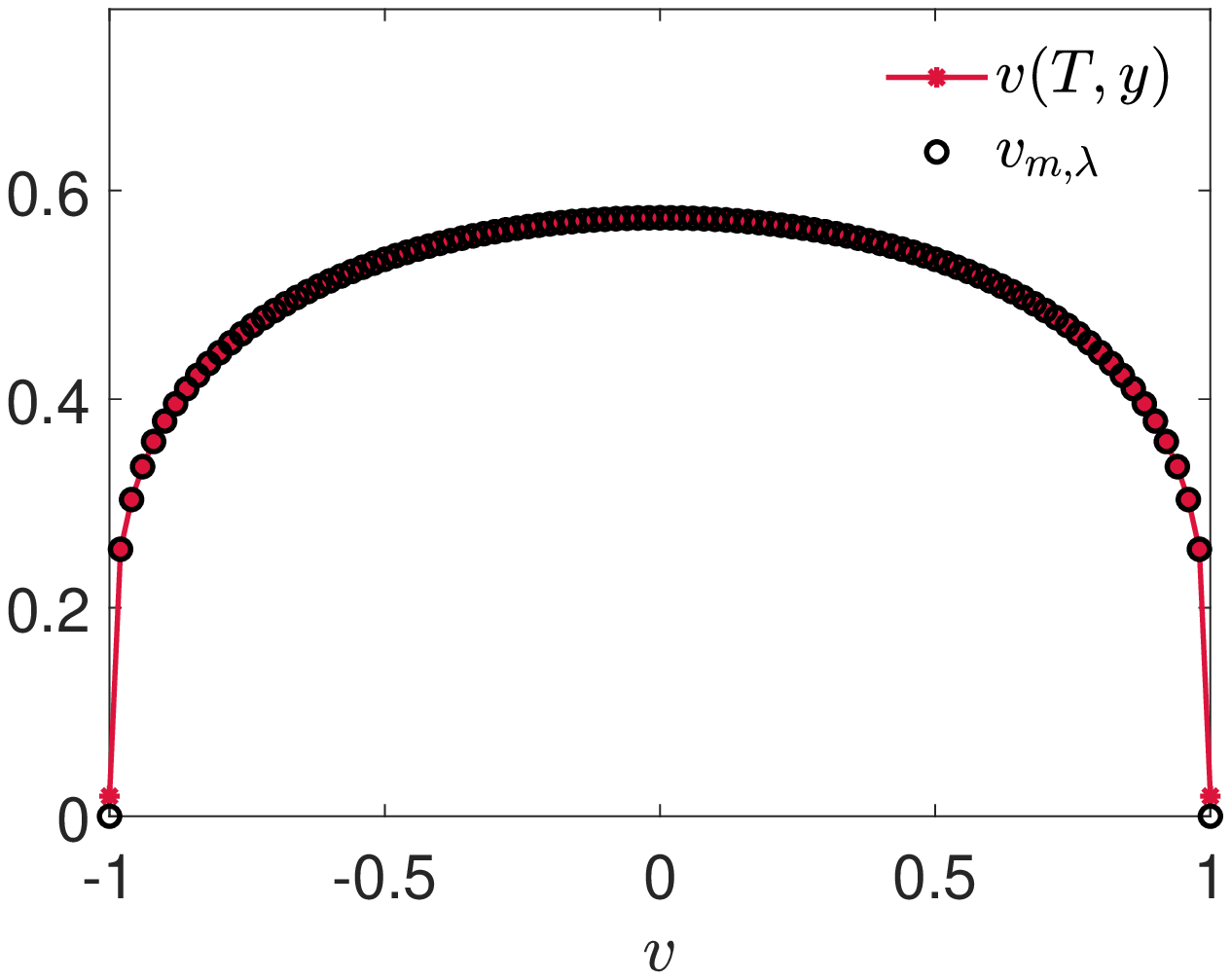}}
    \caption{Comparison between the analytic and numerical steady state solutions of the Fokker--Planck model for different values of the  parameter $\lambda$. Top left $\lambda=0.2$, top right $\lambda=0.4$, bottom left $\lambda=0.6$, bottom right $\lambda=0.8$.  }\label{fig:test2}
\end{figure}

\section{Conclusions}\label{concl}
In this paper, we investigated the large-time behavior of the solution of a Fokker--Planck type equation arising in the study of opinion formation. The same equation, in adjoint form, is well-known under the name of Wright--Fisher equation, and has been exhaustively studied, among others, in a recent paper by Epstein and Mazzeo \cite{EM10} from the point of view of semigroup theory. Our approach to the analysis of the large-time behavior of the solution is different, and relies on the classical study of the evolution of the relative Shannon entropy, which is of common use in the field of kinetic theory. The study of lower bounds for the relative entropy production leads to a new type of logarithmic-Sobolev inequality with weight, satisfied by the Beta-type densities, which allow us in various cases to conclude with exponential convergence to the equilibrium with an explicit rate. 

The case in which the Beta-type density reduces to a uniform variable separates in a natural way from the others, and gives rise to the corresponding of the Euclidean logarithmic-Sobolev inequality. 
\vskip 2cm

\section*{Acknowledgement} This work has been written within the
activities of GNFM and GNAMPA groups  of INdAM (National Institute of
High Mathematics).

The support of the Italian Ministry of Education, University and Research (MIUR) through the 
``Dipartimenti di Eccellenza Program (2018--2022)'' - Dept. of Mathematics ``F. Casorati'', University of Pavia, is kindly acknowledged. 

The authors also kindly acknowledge R. Mazzeo for fruitful explanations on the paper \cite{EM10}, and M. Zanella, who performed the numerical experiments of Section \ref{nume} by means of the entropic numerical scheme introduced in \cite{PZ}.

\bibliography{biblio-wright-fisher}
\bibliographystyle{plain}

\end{document}